\documentclass{article}
\usepackage{fullpage,amssymb,url,amsmath}
\usepackage{graphicx}
\def\sign{\mathrm{sign}}
\def\idf{\mathrm{id}}
\def\KL{\mathrm{KL}}
\def\tr{\mathrm{tr}}
\def\dnu{\mathrm{d}\nu}

\def\calX{\mathcal{X}}
\def\calC{\mathcal{C}}
\def\calD{\mathcal{D}}
\def\calF{\mathcal{F}}
\def\bbR{\mathbb{R}}

\newenvironment{proof}{\paragraph{Proof:}}{\hfill$\square$}
\def\Bhat{\mathrm{Bhat}}
\def\JB{\mathrm{JB}}

\newtheorem{definition}{Definition}
\newtheorem{theorem}{Theorem}
\newtheorem{corollary}{Corollary}
\newtheorem{lemma}{Lemma}
\newtheorem{example}{Example}
\newtheorem{remark}{Remark}

\title{Generalizing Jensen and Bregman divergences with comparative convexity and the statistical Bhattacharyya distances with comparable means\footnote{A Java\texttrademark{} source code for
 reproducible research is available at~\cite{CCCMcode}.
}
}

\date{}

\author{Frank Nielsen\footnote{Computer Science Department LIX, \'Ecole Polytechnique, 91128 Palaiseau Cedex, France and
  Sony Computer Science Laboratories Inc, Tokyo 141-0022, Japan. E-mail:{\tt Frank.Nielsen@acm.org}} \and Richard Nock\footnote{Data61, Sydney, Australia, the Australian National University  \& the University of Sydney, Australia.}}

\begin{document}
\maketitle

\begin{abstract}
Comparative convexity is a generalization of  convexity relying on  abstract notions of means.
We define  the (skew) Jensen divergence and the Jensen diversity from the viewpoint of comparative convexity,
and show how to obtain the generalized Bregman divergences as limit cases of skewed Jensen divergences. 
In particular, we report explicit formula of these generalized Bregman divergences when considering quasi-arithmetic means.
Finally, we introduce a generalization of the Bhattacharyya statistical distances based on comparative means using relative convexity.
\end{abstract}


\section{Introduction: Convexity and comparative convexity}

\subsection{Convexity: Basic definitions}
To start with, let us recall some elementary definitions of convexity, and then introduce the broader notion of {\em comparative convexity}~\cite{ConvexFunction-2006,Hormander2007}.

A set $\calX\subset\bbR^d$ is {\em convex} if and only if (iff):
\begin{equation}
\forall p,q\in\calX, \quad \forall\lambda\in [0,1],\quad (1-\lambda) p+\lambda q\in\calX.
\end{equation}
Geometrically speaking, this means that the line segment $[pq]$ is fully contained inside $\calX$.
 
A  real-valued  continuous function $F$ on a convex domain $\calX$ is said {\em convex} iff:
\begin{equation}\label{eq:Jineq}
\forall p,q\in\calX,\quad \forall\lambda\in [0,1],\quad F((1-\lambda) p+\lambda q) \leq \lambda F(p)+(1-\lambda)F(q).
\end{equation}
A convex function is necessarily continuous~\cite{ConvexFunction-2006}.

It is {\em strictly convex} iff:
\begin{equation}\label{eq:Jsineq}
\forall p,q\in\calX,\quad p\not=q,\quad \forall\lambda\in (0,1),\quad F(\lambda p+(1-\lambda)q) < \lambda F(p)+(1-\lambda)F(q).
\end{equation}

Historically,  Jensen~\cite{Jensen-1905,Jensen-1906} defined in 1905 the notion of convexity using the
 {\em midpoint convexity} property (see (1) of~\cite{Jensen-1906}, page 176): 
\begin{equation}\label{eq:Jmidineq}
F(p)+F(q)\geq 2 F\left(\frac{p+q}{2}\right).
\end{equation}
A function satisfying this Jensen convexity inequality property may not be continuous~\cite{Maksa-2015}. 
But it turns out that for a {\em continuous function} $F$, the midpoint convexity implies the general convexity definition of Eq.~\ref{eq:Jineq}, see~\cite{ConvexFunction-2006}. A continuous and twice differentiable real-valued function $F$ is strictly convex iff $F''(x)>0$.
The well-known generalized criterion for multivariate convex functions consists in checking the positive-definiteness of the Hessian of the function, $\nabla^2 F\succ 0$.
This characterization is due to Alexandrov~\cite{Alexandrov-1939} in 1939.
Let $\calC$ denote the class of strictly convex real-valued functions.

When function $F$ is convex, its {\em epigraph} $\calF=\{(x,y)\ :\ x\in\calX, y\in\bbR,\ F(x)\leq y\}$ of $F$ is a {\em convex object} of $\calX\times \bbR$.
We can interpret geometrically the convexity of Eq.~\ref{eq:Jineq} by noticing that the {\em chord line segment} linking $(p,F(p))$ to $(q,F(q))$ is above the function plot $(x,F(x))$.  
Thus Inequality~\ref{eq:Jineq} can be written more generally as:
\begin{equation}\label{eq:Jineqg}
\forall p,q\in\calX,\quad \forall\lambda,\lambda'\in [0,1],\quad F((1-\lambda) p+\lambda q) \leq \lambda' F(p)+(1-\lambda')F(q).
\end{equation}

\subsection{Comparative convexity}
The notion of convexity can be generalized by observing that in Eq.~\ref{eq:Jmidineq}, rewritten as $\frac{F(p)+F(q)}{2}\geq F\left(\frac{p+q}{2}\right)$, {\em two arithmetic means} are used: One in the {\em domain} of the function ({\it ie.}, $\frac{p+q}{2}$), and one in the {\em codomain} of the function ({\it ie.}, $\frac{F(p)+F(q)}{2}$).
The branch of {\em comparative convexity}~\cite{ConvexFunction-2006} studies classes of $(M,N)$-convex functions $F$ that satisfies the following generalized midpoint convexity inequality:
\begin{equation}\label{eq:JD}
F(M(p,q))\leq N(F(p),F(q)),\quad \forall p,q\in\calX,
\end{equation}
where $M$ and $N$ are two abstract mean functions defined on the domain $\calX$ and codomain $\bbR$, respectively. 
That is, the field of convexity can be defined informally as the study of function behaviors under the actions of means.
This generalization of convexity was first studied by Aumann~\cite{Aumann-1933} in 1933.
Let $\calC_{M,N}$ denote the class of strictly $(M,N)$-convex functions.

There are many kinds of means~\cite{Bullen-2013}.
For example, the well-known {\em Pythagorean means} for  $p,q\in\bbR_{++}=(0,\infty)$ are:
\begin{itemize}
\item the {\em arithmetic mean} (A): $A(p,q)=\frac{p+q}{2}$, 
\item the {\em geometric mean} (G): $G(p,q)=\sqrt{pq}$,  and 
\item the {\em harmonic mean} (H): $H(p,q)=\frac{2}{\frac{1}{p}+\frac{1}{q}}=\frac{2pq}{p+q}$.
\end{itemize}

Thus comparative convexity generalizes the notion of {\em ordinary} convexity that is obtained by choosing $M(x,y)=N(x,y)=A(x,y)$, the arithmetic mean.
Notice that it follows from the Arithmetic Mean-Geometric Mean (AM-GM) inequality:
\begin{equation}
\forall p,q\in (0,\infty), \lambda \in [0,1], p^{1-\lambda}q^\lambda \leq (1-\lambda)p + \lambda q,
\end{equation}
that $(A,G)$-convexity (commonly called, log-convexity) implies $(A,A)$-convexity, but not the converse.
Indeed,  by definition, $F\in \calC_{A,G}$ satisfies the inequality $F(\frac{p+q}{2})\leq \sqrt{F(p)F(q)}$,
 and the AM-GM inequality yields $\sqrt{F(p)F(q)}\leq \frac{F(p)+F(q)}{2}$.
Thus we have by transitivity $F(\frac{p+q}{2})\leq \frac{F(p)+F(q)}{2}$ 
That is,  $F$ is ordinary convex: $F\in\calC$.
Therefore the $(A,G)$-convex functions are a proper subset of the ordinary convex functions: $\calC_{A,G}\subset \calC$.
Similarly, using the Arithmetic-Geometric-Harmonic (AGH) inequalities $A(x,y)\geq G(x,y)\geq H(x,y)$ (with equality iff $x=y$), we have the following function class inclusion relationship: $\calC_{A,H}\subset \calC_{A,G} \subset \calC$.

\subsection{Abstract means and barycenters}

An {\em abstract mean} $M(p,q)$ aggregates two values to produce an intermediate quantity that satisfies the {\em innerness property}~\cite{Bullen-2013}: 
\begin{equation}
\min\{p,q\} \leq M(p,q) \leq \max\{p,q\}.
\end{equation}

To illustrate the richness of abstract bivariate means, let us describe two generic constructions of  mean families:

\begin{description}

\item[Quasi-arithmetic  means.]

The {\em quasi-arithmetic mean} is defined  for a continuous strictly increasing  function $f:I\subset\bbR \rightarrow J\subset\bbR$ as:
\begin{equation}
M_f(p,q) = f^{-1}\left( \frac{f(p)+f(q)}{2}  \right).
\end{equation}
These means are  also called  Kolmogorov-Nagumo-de Finetti means~\cite{Kolmogorov-1930,Nagumo-1930,deFinetti-1931}.
Without loss of generality, we assume strictly increasing functions instead of monotonic functions since $M_{-f}=M_f$.
Indeed, $M_{-f}(p,q)=(-f)^{-1}(-f(M_f(p,q)))$ and $(-f)^{-1} \circ (-f)=\idf$, the identity function.
By choosing $f(x)=x$, $f(x)=\log x$ or $f(x)=\frac{1}{x}$,  we obtain the Pythagorean arithmetic, geometric, and harmonic means, respectively.

Another family of quasi-arithmetic means are the {\em power means} also called {\em H\"older means}~\cite{Holder-1889}:
\begin{equation}
P_\delta(x,y)=\left(\frac{x^\delta+y^\delta}{2}\right)^{\frac{1}{\delta}},
\end{equation}
They are obtained   for $f(x)=x^\delta$  for $\delta\not =0$ with  $I=J=(0,\infty)$, and include in the limit cases the maximum and minimum values:
$\lim_{\delta\rightarrow\infty} P_\delta(a,b)=\max\{a,b\}$ and $\lim_{\delta\rightarrow -\infty} P_\delta(a,b)=\min\{a,b\}$.
The harmonic mean is obtained for $\delta=-1$: $H=P_{-\delta}$ and the quadratic mean $Q(p,q)=\sqrt{\frac{p^2+q^2}{2}}=P_2$ for $\delta=2$.
To get a smooth family of H\"older means, we define $P_0(x,y)=\sqrt{xy}$, the geometric mean, for $\delta=0$.
The power means are provably the only {\em homogeneous} quasi-arithmetic means: $M_\delta(\lambda a,\lambda b)=\lambda M_\delta(a,b)$ for any $\lambda\geq 0$, see Proposition 3 of~\cite{Pasteczka-2015}.
We refer the Reader to Appendix~\ref{sec:aqam} for an axiomatization of these quasi-arithmetic means due to Kolmogorov~\cite{Kolmogorov-1930} in 1930, and an extension to define {\em quasi-arithmetic expected values} of a random variable.

\item[Lagrange means.]
Lagrange means~\cite{LagrangianMeans-1998} (also termed Lagrangean means) are mean values derived from the mean value theorem.
Assume without loss of generality that $p < q$ so that the mean $m\in [p,q]$.
From the mean value theorem, we have for a differentiable function $f$: 
\begin{equation}
\exists \lambda \in [p,q] :  f'(\lambda) = \frac{f(q)-f(p)}{q-p}.
\end{equation}
Thus when $f'$ is a monotonic function, its inverse function  $f^{-1}$ is well-defined, and the unique mean value mean $\lambda\in [p,q]$ can be defined as:
\begin{equation}
L_f(p,q)= \lambda = (f')^{-1}\left( \frac{f(q)-f(p)}{q-p} \right).
\end{equation}

For example, letting $f(x)=\log(x)$ and $f'(x)=(f')^{-1}(x)=\frac{1}{x}$, we recover the {\em logarithmic mean} (L), that is {\em not} a quasi-arithmetic mean:
$$
\begin{array}{ll}
m(p,q)
&=
\begin{cases}
0 & \text{if } p=0 \text{ or } q=0 ,\\
x & \text{if } p=q ,\\
\frac{q - p}{\log q - \log p} & \text{otherwise,}
\end{cases}
\end{array}
$$

The logarithmic mean is bounded below by the geometric mean and above by the arithmetic mean: 
$G(p,q)\leq L(p,q)\leq A(p,q)$. 

\end{description}

Both quasi-arithmetic and Lagrange mean generators are defined up to an affine term $ax+b$ with $a\not =0$.
Moreover, the intersection of the class of quasi-arithmetic means with the Lagrangean means has been fully characterized in~\cite{Pales-2011}, and include the arithmetic mean $A$.

In general, a mean is {\em strict} when $M(p,q)\in (p,q)$ for $p\not =q$, and {\em symmetric} when $M(p,q)=M(q,p)$.

Yet another interesting  family of means are the {\em Stolarsky means} (S).
The  Stolarsky means  are not quasi-arithmetic means nor mean-value means, and are defined as follows
\begin{equation}
S_p(x,y)= \left(\frac{x^p-y^p}{p(x-y)}\right)^{\frac{1}{p-1}},\quad p\not\in \{0,1\}.
\end{equation}

In limit cases, the  Stolarsky  family of means yields the logarithmic mean (L) when $p\rightarrow 0$ and the {\em identric mean} (I) when $p\rightarrow 1$:
\begin{equation}
I(x,y) =  \left( \frac{y^y}{x^x} \right)^{\frac{1}{y-x}}.
\end{equation}

The Stolarsky means belong to the family of {\em Cauchy mean-value means}~\cite{Bullen-2013} defined for two positive differentiable and strictly  monotonic functions $f$ and $g$ such that $\frac{f'}{g'}$ has an inverse function.
The Cauchy mean-value mean is defined by:
\begin{equation}
C_{f,g}(p,q) = \left(\frac{f'}{g'}\right)^{-1} \left( \frac{f(q)-f(p)}{g(q)-g(p)} \right), \quad q\not= p,
\end{equation}
with $C_{f,g}(p,p) = p$.

The Cauchy means can be reinterpreted as Lagrange means~\cite{weightedCauchy-2006} by the following identity: 
$C_{f,g}(p,q) = L_{f\circ g^{-1}}(g(p),g(q))$ since $((f\circ g^{-1})(x))'=\frac{f'(g^{-1}(x))}{g'(g^{-1}(x))}$:

\begin{eqnarray}
L_{f\circ g^{-1}}(g(p),g(q)) &=& ((f\circ g^{-1})'(g(x)))^{-1}\left( \frac{(f\circ g^{-1})(g(q))- (f\circ g^{-1})(g(p))}{g(q)-g(p)}\right),\\
&=& \left(\frac{f'(g^{-1}(g(x)))}{g'(g^{-1}(g(x)))}\right)^{-1}\left( \frac{f(q)-f(p)}{g(q)-g(p)} \right),\\
&=& C_{f,g}(p,q).
\end{eqnarray}

More generally, we may weight the values and consider {\em barycentric means} $M(p,q;1-\alpha,\alpha)=M_\alpha(p,q)$ for $\alpha\in [0,1]$.
Those {\em weighted means} further satisfy the following smooth interpolation property:
\begin{equation}
M_0(p,q)=p, \quad M_1(p,q)=q, \quad M_{1-\alpha}(p,q)=M_{\alpha}(q,p).
\end{equation}

For example, a {\em quasi-arithmetic barycentric mean} is defined for a monotone function $f$ by:
\begin{equation}
M_f(p,q;1-\alpha,\alpha)= M_{f,\alpha}(p,q) = f^{-1}\left( (1-\alpha)f(p)+ \alpha f(q)   \right).
\end{equation}

\begin{definition}[Regular mean]
A mean is said {\em regular} if it is:
\begin{enumerate}
\item homogeneous,
\item symmetric, 
\item continuous, and 
\item increasing in each variable.
\end{enumerate}
\end{definition}

In this paper, we shall consider regular means and weighted means.
The Pythagorean means are regular means that can be extended to {\em Pythagorean barycenters} (weighted means) for  $p,q\in\bbR_{++}$ as follows:
\begin{itemize}
\item the {\em arithmetic barycenter} (A): $A(p,q;1-\alpha,\alpha)=(1-\alpha)p+\alpha q$, 
\item the {\em geometric barycenter} (G): $G(p,q;1-\alpha,\alpha)= p^{1-\alpha}q^\alpha$,  and 
\item the {\em harmonic barycenter} (H): $H(p,q;1-\alpha,\alpha)= \frac{1}{(1-\alpha)\frac{1}{p}+\alpha\frac{1}{q}}=\frac{pq}{\alpha p+(1-\alpha)q}$.
\end{itemize}

The power barycenters (P) are defined by $P_\delta(p,q;1-\alpha,\alpha)=\left((1-\alpha)p^\delta+\alpha q^\delta\right)^{\frac{1}{\delta}}$ for $\delta\not =0$.
Those power barycenters  generalize the arithmetic and harmonic barycenters, and can be extended into a smooth family of barycenters 
by setting $P_0(p,q;1-\alpha,\alpha)=G(p,q;1-\alpha,\alpha)$. 

Let us give two families of  means that are not quasi-arithmetic means:
The weighted {\em Lehmer mean}~\cite{Beliakov-2015} of order $\delta$ is defined for $\delta\in\bbR$ as:
\begin{equation}
L_\delta(x_1,\ldots,x_n;w_1,\ldots,w_n) = \frac{\sum_{i=1}^n w_i x_i^{\delta+1}}{\sum_{i=1}^n w_i x_i^\delta}.
\end{equation}

Notice that we have $L_{-\frac{1}{2}}=G$ (the geometric mean) since the denominator of $L_{-\frac{1}{2}}(p,q)$ rewrites as
 $p^{-\frac{1}{2}}+q^{-\frac{1}{2}}=\frac{\sqrt{pq}}{\sqrt{p}+\sqrt{q}}$.
The Lehmer means intersect with the H\"older means only for the arithmetic, geometric and harmonic means.
However the Lehmer mean $L_2$ is not a regular mean since it is {\em not} increasing in each variable.
Indeed, $L_1(x,y)=\frac{x^2+y^2}{x+y}=C(x,y)$ is the contraharmonic mean.

The family of  Lehmer barycentric means can further be encapsulated into the family of {\em Gini means}:
\begin{equation}
G_{\delta_1,\delta_2}(x_1,\ldots,x_n;w_1,\ldots,w_n) =
\left\{
\begin{array}{ll}
\left( \frac{\sum_{i=1}^n w_i x_i^{\delta_1}}{\sum_{i=1}^n w_i x_i^{\delta_2}}\right)^{\frac{1}{\delta_1-\delta_2}} & \delta_1\not=\delta_2,\\
\left(\prod_{i=1}^n x_i^{w_i x_i^\delta} \right)^{\frac{1}{\sum_{i=1}^n w_i x_i^\delta} } & \delta_1=\delta_2=\delta.
\end{array}
\right.
\end{equation}

Those families of Gini and Lehmer means are homogeneous means:
$G_{\delta_1,\delta_2}(\lambda x_1,\ldots,\lambda x_n;w_1,\ldots,w_n)= \lambda G_{\delta_1,\delta_2}(x_1,\ldots,x_n;w_1,\ldots,w_n)$ for any $\lambda>0$.
The family of Gini means include the power means: $G_{0,\delta}=P_\delta$ for $\delta\leq 0$ and $G_{\delta,0}=P_\delta$ for $\delta\geq 0$.

The Bajraktarevic means~\cite{Beliakov-2014} are also not regular.

Given a symmetric and homogeneous mean $M(x,y)$, we can associate a dual mean $M^*(x,y)=\frac{1}{M(\frac{1}{x},\frac{1}{y})}=\frac{xy}{M(x,y)}$ that is symmetric, homogeneous, and satisfies $(M^*)^*=M$.  We write concisely $M^*=\frac{G^2}{M}$ (the geometric mean $G$ is self-dual), and we have $\min^*=\max$ and $\max^*=\min$.

\subsection{Paper outline}
The goal of this paper is to generalize the ordinary Jensen, Bregman and Bhattachayya distances~\cite{BR-2011} using an extended notion of convexity.
In particular, the classes of generalized convex functions $\calC_{M,N}$ generalize the ordinary convex functions (the standard $(A,A)$-convexity), and include the following classes:

\begin{itemize}

	\item the class of {\em log-convex functions} ($M$ the arithmetic mean and $N$ the geometric mean),
	
	\item the class of {\em multiplicatively convex functions} ($M$ and $N$ both geometric
means),

\item the class of {\em $M_p$-convex functions} ($M$ the arithmetic mean and $N$ the $p$-th power mean).
	
\end{itemize}

The paper is organized as follows:
Section~\ref{sec:GenCCD} defines the generalized Jensen and skew Jensen divergences from generalized convexity inequalities, extends the definitions to Jensen diversity, and
 introduces the generalized Bregman divergences as a limit case of skew Jensen divergences.
Section~\ref{sec:QAB} considers the class of quasi-arithmetic means to report explicit formulas for these generalized Bregman divergences.
In Section~\ref{sec:GB}, we introduce a generalization of the statistical Bhattacharyya  divergence and  of the Bhattacharyya coefficient using the concept of comparable means, and show how to obtain closed-form expressions by adapting the means to the structure of the input distributions.
Finally, Section~\ref{sec:concl} concludes and hints at further perspectives.
For sake of completeness, the axiomatization of quasi-arithmetic means are reported in Appendix~\ref{sec:aqam}.

\section{Generalized Jensen, skewed Jensen and Bregman divergences}\label{sec:GenCCD}

The Jensen midpoint inequality of Eq.~\ref{eq:Jmidineq} can be used to build a symmetric dissimilarity measure $J_F(p,q)=J_F(q,p)$ for $F\in\calC$, originally called the Jensen difference   in~\cite{JensenDiversity-1982}:

\begin{equation}
J_F(p,q) = \frac{F(p)+F(q)}{2} - F\left(\frac{p+q}{2}\right) \geq 0.
\end{equation}

Nowadays, that distance is called a {\em Jensen divergence} (or a Burbea-Rao divergence~\cite{BR-2011}).
The term ``divergence'' is traditionally used in information geometry~\cite{IG-2016} instead  of distance to emphasize the fact that the dissimilarity may not be a metric.
A divergence $D(p,q)$ only requires:
\begin{enumerate} 
\item to satisfy the law of the indiscernible

\begin{equation}
D(p,q) = 0 \Leftrightarrow p=q,
\end{equation}
and
\item to be thrice differentiable in order define a differential-geometric structure involving a metric tensor and a cubic tensor~\cite{IG-2016}.
\end{enumerate}

It follows by construction from the Jensen inequality that $J_F(p,q) \geq 0$, and that $J_F(p,q)=0$ iff $p=q$ for a {\em strictly convex} function $F$.

\subsection{Jensen Comparative Convexity Divergences}

Let us extend the definitions of Jensen, skewed Jensen   divergences to the setting of comparative convexity as follows:

\begin{definition}[Jensen Comparative Convexity Divergence, JCCD]
The Jensen Comparative Convexity Divergence (JCCD) is defined for a strictly $(M,N)$-convex function $F\in\calC_{M,N}:I\rightarrow \bbR$  by:

\begin{equation}\label{eq:JCCD}
\boxed{J_F^{M,N}(p,q) = N(F(p),F(q)))-F(M(p,q))}
\end{equation}
\end{definition}

For symmetric means $M$ and $N$, the JCCD is a symmetric divergence: $J_F^{M,N}(p,q)=J_F^{M,N}(q,p)$.
It follows from the strict $(M,N)$-convexity property of $F$ that $J_F^{M,N}(p,q)=0$ iff $p=q$.

The definition of the JCCD can be extended to skew JCCDs by taking the barycentric {\em regular} means:

\begin{definition}[Skew Jensen Comparative Convexity Divergence]\label{def:jccd}
The skew $\alpha$-Jensen Comparative Convexity Divergence (sJCCD) is defined for a strictly
 $(M,N)$-convex function $F\in\calC_{M,N}:I\rightarrow \bbR$ where $M$ and $N$ are regular means and $\alpha\in (0,1)$  by:

\begin{equation}\label{eq:sJCCD}
\boxed{J_{F,\alpha}^{M,N}(p:q) = N_\alpha(F(p),F(q))-F(M_\alpha(p,q)).}
\end{equation}
\end{definition}

It follows that $J_{F,1-\alpha}^{M,N}(p:q)=J_{F,\alpha}^{M,N}(q,p)$.
The fact that $J_{F,\alpha}^{M,N}(p:q)\geq 0$ follows from the midpoint $(M,N)$-convexity property of function $F$
 (see Theorem~A  page 4 and Section 2.6 page 88 of~\cite{ConvexFunction-2006}).
In fact, the generalized midpoint convexity inequality plus the continuity assumption yields an exact characterization of $(M,N)$-convex functions, see~\cite{ConvexFunction-2006}.
The power means (including the harmonic mean, the arithmetic mean and the geometric mean by extension)  are examples of regular means.
Note that the exponential function $\exp(x)$ is both $(L,L)$-convex and $(I,I)$-convex, two logarithmic and identric regular Stolarsky's means.

In some cases, when the barycentric means is well-defined for $\alpha\in\bbR$ ({\it ie.}, extrapolating values when $\alpha<0$ or $\alpha>1$), we can extend the skew Jensen divergences to $\alpha\in \bbR\backslash\{0,1\}$. 
For example, using the arithmetic means $M=N=A$, we may define for $\alpha\in \bbR\backslash\{0,1\}$:
\begin{equation}\label{eq:sJCCDA}
J_{F,\alpha}(p:q) =  \sign(\alpha(1-\alpha)) \left( A_\alpha(F(p),F(q))-F(A_\alpha(p,q))\right),
\end{equation}
where $A_\alpha(p,q)=(1-\alpha)p+\alpha q$.

\begin{example}[Jensen divergence for multiplicatively convex functions]
The class $\calC_{G,G}$ of strictly $(G,G)$-convex functions is called the class of {\em multiplicatively convex functions}.
Let $F\in\calC_{G,G}$. We get the $(G,G)$-Jensen divergence:
\begin{equation}
J^{G,G}_{\exp}(p,q) = \sqrt{F(p)F(q)}-F(\sqrt{pq}) \geq 0.
\end{equation}

We check that $J^{G,G}(x,x)=0$.
It turns out that $F\in\calC_{G,G}$ when $\log F (x)$
is a convex function of $\log x$ (see Lemma~\ref{lemma:cvx}).
Some examples of multiplicatively convex functions are $F(x)=\exp(x)$, $F(x)=\sinh(x)$, $F(x)=\Gamma(x)$ (the Gamma function generalizing the factorial), $F(x)=\exp(\log^2 x)$.
For example, take $F(x)=\exp(x)$.
Then the corresponding $(G,G)$-JCCD is:
\begin{equation}
J^{G,G}(p:q)  = \exp\left(\frac{p+q}{2}\right) - \exp(\sqrt{pq})\geq 0.
\end{equation}
\end{example}

\subsection{Jensen  Comparative Convexity Diversity Indices}
The $2$-point divergences ({\it ie.}, dissimilarity measure between two points) can be extended to a positively weighted set of values by defining a notion of {\em diversity}~\cite{JensenDiversity-1982}  as:

\begin{definition}[Jensen  Comparative Convexity Diversity Index, JCCDI]
Let $\{(w_i,x_i)\}_{i=1}^n$ be a set of $n$ positive weighted values so that $\sum w_i=1$.
Then the  {\em Jensen diversity index} with respect to the strict $(M,N)$-convexity of a function $F$ is:
\begin{equation}\label{eq:JDI}
\boxed{J_F^{M,N}(x_1,\ldots, x_n; w_1, \ldots, w_n)   = N(F(x_1),\ldots, F(x_n); w_1, \ldots, w_n) - F (M(x_1,\ldots, x_n; w_1, \ldots, w_n))}
\end{equation}
\end{definition}

It is proved in~\cite{Niculescu-2003} that $N(F(x_1),\ldots, F(x_n); w_1, \ldots, w_n) \geq F (M(x_1,\ldots, x_n; w_1, \ldots, w_n))$ for a continuous $(M,N)$-convex function. Therefore, we have $J_F^{M,N}(x_1,\ldots, x_n; w_1, \ldots, w_n)  \geq 0$.
See also  Theorem~A  page 4  of~\cite{ConvexFunction-2006}.

When both means $M$ and $N$ are set to the arithmetic mean, this diversity index has been called the {\em Bregman information}~\cite{BD-2005} in the context of clustering with Bregman divergences. The Bregman information generalizes the notion of variance of a cluster obtained for the generator $F(x)=x^\top x$.

\subsection{Bregman Comparative Convexity Divergences}

Let us define the {\em Bregman Comparative Convexity Divergence} (BCCD), also called generalized $(M,N)$-Bregman divergence, as the limit case of skew JCCDs:

\begin{definition}[Bregman Comparative Convexity Divergence, BCCD]\label{def:bccd}
The Bregman Comparative Convexity Divergence (BCCD) is defined for a strictly $(M,N)$-convex function $F:I\rightarrow \bbR$   by

\begin{equation}\label{eq:BCCD}
\boxed{
B_{F}^{M,N}(p:q) = \lim_{\alpha\rightarrow 1^-}   \frac{1}{\alpha(1-\alpha)}J_{F,\alpha}^{M,N}(p:q) = \lim_{\alpha\rightarrow 1^-}  \frac{1}{\alpha(1-\alpha)} \left( N_\alpha(F(p),F(q)))-F(M_\alpha(p,q)) \right)
}
\end{equation}
\end{definition}

It follows from the symmetry $J_{F,\alpha}(p:q)=J_{F,1-\alpha}(q:p)$ (for symmetric means) that when the limits exists, we get the {\em reverse Bregman divergence}:
\begin{equation}
B_{F}^{M,N}(q:p)= \lim_{\alpha\rightarrow 0^+}   \frac{1}{\alpha(1-\alpha)}J_{F,\alpha}^{M,N}(p:q)  = \lim_{\alpha\rightarrow 0^+} \frac{1}{\alpha(1-\alpha)} \left( N_\alpha(F(p),F(q)))-F(M_\alpha(p,q)) \right).
\end{equation}

Note that the limits are one-sided limits.

Notice that when both means $M$ and $N$ are chosen as the arithmetic mean, we recover the ordinary Jensen, skew Jensen and Bregman divergences described and studied in~\cite{BR-2011}.
This generalization of Bregman divergences has also been studied by Petz~\cite{Petz-2007} to get generalized quantum relative entropies.
Petz defined the Bregman divergence between two points $p$ and $q$ of a convex set $C$ in a Banach space for a
 given function $F:C\rightarrow \mathcal{B}(\mathcal{H})$ (Banach space induced by a Hilbert space $\mathcal{H}$) as:
\begin{equation}
B_F(p:q)=F(p)-F(q)-\lim_{\alpha\rightarrow 0^+} \frac{1}{\alpha}(F(q+\alpha(p-q)) -F(q)).
\end{equation}
Indeed, this last equation can be rewritten as:
\begin{eqnarray}
B_F(p:q) &=& \lim_{\alpha\rightarrow 0^+} \frac{1}{\alpha} (\alpha F(p)-(1-\alpha)F(q) - (F(q+\alpha(p-q))),\\
&=&  \lim_{\alpha\rightarrow 1^-} \frac{1}{1-\alpha} (A_{\alpha}(F(p),F(q)) -F(A_{\alpha}(p,q))),\\
&=&   \lim_{\alpha\rightarrow 1^-} \frac{1}{1-\alpha} J_{F,\alpha}^{A,A}(p,q).
\end{eqnarray}
When $C$ is the set of positive semi-definite matrices of unit trace and $F(x)=x\log x$, then the induced Bregman divergence
is Umegaki's relative entropy~\cite{Petz-2007}: $B_F(p:q)=\tr p(\log p -\log q)$.

Thus we have a general recipe to get generalized Bregman divergences:
Study the asymptotic barycentric symmetric mean expansions of $M(p,q;1-\alpha,\alpha)$ and $N(F(p),F(q);1-\alpha,\alpha)$ when $\alpha\rightarrow 0$, and deduce the generalized $(M,N)$-Bregman divergence provided the limits of $\frac{1}{\alpha}M(p,q;1-\alpha,\alpha)$ and $\frac{1}{\alpha}N(F(p),F(q);1-\alpha,\alpha)$ 
exist when $\alpha\rightarrow 0$.

Letting $\omega=2\alpha-1\in(-1,1)$ (or $\alpha=\frac{1+\omega}{2}\in(0,1)$) and using barycentric means, we can define the following divergence:

\begin{equation}
\boxed{D_{F,\omega}^{M,N}(p:q) =  \frac{1}{1-\omega^2}\left( N\left(F(p),F(q);\frac{1-\omega}{2};\frac{1+\omega}{2}\right) - 
F\left( M\left(p,q;\frac{1-\omega}{2};\frac{1+\omega}{2}\right)\right) \right)}
\end{equation}

Then the generalized Bregman divergences are obtained in the limit cases when $\omega\rightarrow\pm 1$.
Notice that in~\cite{Zhang-2004} (Sec. 3.5), Zhang defined a divergence functional from generalized (quasi-arithmetic) means:
For $f$ a strictly convex and strictly monotone increasing function, Zhang defined the divergence for $\rho=f^{-1}$ as follows:

\begin{equation}
\calD^{(\alpha)}_\rho(p,q)=\frac{4}{1-\alpha^2}\int_\calX \left(
\frac{1-\alpha}{2}p + \frac{1+\alpha}{2}q - M_\rho(p,q) 
\right) \dnu(x).
\end{equation}

We shall investigate such generalized $(M,N)$-Bregman divergences when both means are weighted quasi-arithmetic means in Section~\ref{sec:QAB}.

\subsection{From univariate to multivariate separable divergences}
Multivariate divergences can be built from univariate divergences component-wise.
For example, let $P=(P_1,\ldots, P_d)$ and $Q=(Q_1,\ldots, Q_d)$ be two vectors of $\bbR^d$, and consider the following multivariate generalized Bregman divergence:
\begin{equation}
B_F(P:Q) = \sum_{i=1^d} B_{F_i}^{M_i,N_i}(P_i:Q_i),
\end{equation}
where $F_i\in\calC_{M_i,N_i}$ is a $(M_i,N_i)$-convex function.
These divergences  can be decomposed as a sum of univariate divergences, and are thus  called {\em separable divergences} in the literature~\cite{LearningDivergence-2015}.

\begin{remark}
Observe that the BCCD can be approximated in practice from the JCCD by taking small values for $\alpha>0$:
For example, the ordinary Bregman divergence can be approximated from the ordinary skew Jensen divergence as follows:

\begin{equation}
B_F(q:p) \simeq \frac{1}{\alpha(1-\alpha)} \left ((1-\alpha)F(p)+\alpha F(q) -F((1-\alpha) p + \alpha q)\right),\quad \mbox{$\alpha>0$ small}.
\end{equation}
This is all the more interesting in practice for approximating the Bregman divergence by skipping the calculation of the gradient $\nabla F$.
\end{remark}

We shall now report explicit formulas for the generalized Bregman divergences when using quasi-arithmetic means.

\section{Quasi-arithmetic Bregman divergences}\label{sec:QAB}

Let us report direct formulas for the generalized Bregman divergences defined with respect to quasi-arithmetic comparative convexity.
Let $\rho$ and $\tau$ be two continuous differentiable functions defining the quasi-arithmetic means $M_\rho$ and $M_\tau$, respectively.

\subsection{A direct formula}
By definition, a function $F\in\calC_{\rho,\tau}$ is $(M_\rho,M_\tau)$-convex iff:
\begin{equation} 
M_\tau(F(p),F(q))) \geq F(M_\rho(p,q)).
\end{equation}
This  $(M_\rho,M_\tau)$-midpoint convexity property with the continuity of $F$ yields the more general definition of $(M_\rho,M_\tau)$-convexity:
\begin{equation} 
M_{\tau,\alpha}(F(p),F(q))) \geq F(M_{\rho,\alpha}(p,q)),\quad  \alpha\in [0,1].
\end{equation}

Let us study those quasi-arithmetic Bregman divergences $B^{\rho,\tau}_F$ obtained when taking the limit:
\begin{equation} 
B^{\rho,\tau}_F(q:p) = \lim_{\alpha\rightarrow 0} \frac{1}{\alpha(1-\alpha)} \left( M_{\tau,\alpha}(F(p),F(q)))- F\left(M_{\rho,\alpha}(p,q)\right) \right),
\end{equation}
for $M_{\rho,\alpha}$ and $M_{\tau,\alpha}$ two quasi-arithmetic barycentric means obtained for continuous and monotonic functions $\rho$ and $\tau$, respectively.
Recall that a quasi-arithmetic barycentric mean for a monotone function $\tau$ is defined by:
\begin{equation}
M_{\alpha}(p;q)= \tau^{-1}\left(\tau(p)+\alpha(\tau(q)-\tau(p))\right), \quad \alpha\in [0,1],\quad M_{0}(p;q)=p, M_{1}(p;q)=q.
\end{equation}

We state the generalized Bregman divergence formula obtained with respect to quasi-arithmetic comparative convexity:

\begin{theorem}[Quasi-arithmetic Bregman divergences, QABD\label{theo:qabd}]
Let $F:I\subset \bbR \rightarrow \bbR$ be a real-valued $(M_\rho,M_\tau)$-convex function defined on an interval $I$ for two strictly monotone and differentiable functions $\rho$ and $\tau$.
The quasi-arithmetic Bregman divergence (QABD) induced by the comparative convexity is:
\begin{equation}
\boxed{
B^{\rho,\tau}_F(p:q) 
=  
\frac{\tau(F(p))-\tau(F(q))}{\tau'(F(q))} - \frac{\rho(p)-\rho(q)}{\rho'(q)} F'(q).
}
\end{equation}
\end{theorem}

\begin{proof} 

By taking the first-order Taylor expansion of $\tau^{-1}(x)$ at $x_0$, we get:
\begin{equation}
\tau^{-1}(x)\simeq_{x_0} \tau^{-1}(x_0) + (x-x_0) (\tau^{-1})'(x_0).
\end{equation}
Using the property of the derivative of an inverse function:
\begin{equation}
(\tau^{-1})'(x)=\frac{1}{(\tau'(\tau^{-1})(x))},
\end{equation} 
it follows that the first-order Taylor expansion of $\tau^{-1}(x)$ is:
\begin{equation}
\tau^{-1}(x)\simeq \tau^{-1}(x_0)+ (x-x_0) \frac{1}{(\tau'(\tau^{-1})(x_0))}.
\end{equation}

Plugging $x_0=\tau(p)$ and $x=\tau(p)+\alpha(\tau(q)-\tau(p))$, we get a first-order approximation of the barycentric quasi-arithmetic mean $M_\tau$ when $\alpha\rightarrow 0$:

\begin{equation}
M_\alpha(p,q)  \simeq  p +  \frac{\alpha(\tau(q)-\tau(p))}{\tau'(p)}.
\end{equation}

For example, when $\tau(x)=x$ ({\it ie.}, arithmetic mean), we have $A_\alpha(p,q)  \simeq  p+\alpha (q-p)$,
when $\tau(x)=\log x$ ({\it ie.}, geometric mean), we obtain $G_\alpha(p,q)  \simeq  p+\alpha p \log\frac{q}{p}$, and when $\tau(x)=\frac{1}{x}$ ({\it ie.}, harmonic mean) we get $H_\alpha(p,q)  \simeq  p+\alpha(p-\frac{p^2}{q})$.
For the regular power means, we have $P_\alpha(p,q)  \simeq  p+\alpha \frac{q^\delta-p^\delta}{\delta p^{\delta-1}}$.
These are first-order weighted mean approximations obtained by small values of $\alpha$.
Suppose $p < q$. Since $\tau$ is a monotone function: When $\tau$ is strictly increasing, $\tau'>0$ and $\tau(q)>\tau(p)$.
Therefore  $\frac{\tau(q)-\tau(p)}{\tau'(p)}>0$. Similarly, when $\tau$ is strictly decreasing, $\tau'<0$ and $\tau(q)<\tau(p)$ so that $\frac{\tau(q)-\tau(p)}{\tau'(p)}>0$.

Now, consider the skewed Jensen Comparative Convexity Distance defined by:
\begin{equation}
J^{\tau,\rho}_{F,\alpha}(p:q)=   (M_{\tau,\alpha}(F(p),F(q)) - F(M_{\rho,\alpha}(p,q) ) ),
\end{equation}
and apply a first-order Taylor expansion to get:
\begin{equation}
F(M_{\tau,\alpha}(p,q) ))\simeq  F\left (p +  \frac{\alpha(\tau(q)-\tau(p))}{\tau'(p)} \right)
\simeq F(p)+ \frac{\alpha(\tau(q)-\tau(p))}{\tau'(p)} F'(p)
\end{equation}

Thus it follows that the Bregman divergence for quasi-arithmetic comparative convexity is:

\begin{equation}
B^{\rho,\tau}_F(q:p) = \lim_{\alpha\rightarrow 0} \frac{1}{\alpha(1-\alpha)}J_{\tau,\rho,\alpha}(p:q) =  
\frac{\tau(F(q))-\tau(F(p))}{\tau'(F(p))} - \frac{\rho(q)-\rho(p)}{\rho'(p)} F'(p),
\end{equation}

and the reverse Bregman divergence is:

\begin{equation}
B^{\rho,\tau}_F(p:q) = \lim_{\alpha\rightarrow 1} \frac{1}{\alpha(1-\alpha)} J^{\tau,\rho}_\alpha(p:q) =  \lim_{\alpha\rightarrow 0} \frac{1}{\alpha(1-\alpha)} J^{\tau,\rho}_{\alpha}(q:p) 
\end{equation}

For notational convenience, let us define the following {\em auxiliary function}:
\begin{equation}\label{eq:BDkappa}
\kappa_\gamma(x:y) =   \frac{\gamma(y)-\gamma(x)}{\gamma'(x)}.
\end{equation} 

Then the generalized Bregman divergence is written  compactly as:

\begin{equation}
\boxed{B^{\rho,\tau}_F(p:q) = \kappa_\tau(F(q):F(p))  -\kappa_\rho(q:p) F'(q).}
\end{equation}

Table~\ref{tab:kappa} reports the auxiliary function instantiated for usual quasi-arithmetic generator functions.

Since power means are regular means, we get the following family of
 {\em power mean Bregman divergences} for $\delta_1,\delta_2\in\bbR\backslash\{0\}$ with $F\in\calC_{P_{\delta_1},P_{\delta_2}}$:

\begin{equation}
\boxed{B^{\delta_1,\delta_2}_F(p:q) = \frac{F^{\delta_2}(p)-F^{\delta_2}(q)}{\delta_2 F^{\delta_2-1}(q)} - \frac{p^{\delta_1}-q^{\delta_1}}{\delta_1 q^{\delta_1-1}}F'(q) }
\end{equation}

\begin{table}
\centering
$$
\begin{array}{l|l|l}
 \mbox{Type} & \gamma & \kappa_\gamma(x:y)=   \frac{\gamma(y)-\gamma(x)}{\gamma'(x)} \\ \hline
A  & \gamma(x)=x   &  y-x\\
G  &\gamma(x)=\log x  & x\log\frac{y}{x} \\
H & \gamma(x)=\frac{1}{x}  & x^2\left(\frac{1}{y}-\frac{1}{x}\right)\\
P_\delta, \delta\not =0  &\gamma_\delta(x)=x^\delta   & \frac{y^\delta-x^\delta}{\delta x^{\delta-1}} \\
\end{array}
$$

\caption{The auxiliary function $\kappa$ instantiated for the arithmetic, geometric and harmonic generators.
The generalized Bregman divergences write $B^{\rho,\tau}_F(p:q) = \kappa_\tau(F(q):F(p))  -\kappa_\rho(q:p) F'(q)$ for $F$ a real-valued $(M_\rho,M_\tau)$-convex generator. \label{tab:kappa}}
\end{table}

Note that when $\rho(x)=\tau(x)=x$ ({\it ie.}, quasi-arithmetic means yielding arithmetic means), we recover the fact that the skew Jensen difference tends to a Bregman divergence~\cite{BR-2011}:

\begin{equation}
\lim_{\alpha\rightarrow 0}  \frac{1}{\alpha}J_{F,\alpha}(p:q) = B_F(q:p) = F(q)-F(p)-(q-p)F'(p),
\end{equation}
and
\begin{equation}
\lim_{\alpha\rightarrow 1}  \frac{1}{1-\alpha} J_{F,\alpha}(p:q) = B_F(p:q) = F(p)-F(q)-(p-q)F'(q).
\end{equation}

Notice that we required function generator $F$ to be strictly $(M_\rho,M_\tau)$-convex and functions $\rho, \tau$  and $F$ to be differentiable  in order to perform the various Taylor first-order expansions.
\end{proof}

In~\cite{BR-2011}, the Jensen divergence was interpreted as a {\em Jensen-Bregman divergence} defined by:
\begin{equation}
\JB_F(p,q)=\frac{B_F\left(p:\frac{p+q}{2}\right)+B_F\left(q:\frac{p+q}{2}\right)}{2}=\JB_F(q,p).
\end{equation}
The (discrete) {\em Jensen-Shannon divergence}~\cite{Lin-1991} is a Jensen-Bregman divergence for the Shannon information function $F(x)=\sum_{i=1}^d x_i \log x_i$, the negative Shannon entropy: $F(x)=-H(x)$.

It turns out that $\JB_F(p,q)=J_F(p,q)$.
This identity comes from the fact that the terms $p-\frac{p+q}{2}=\frac{p-q}{2}$ and $q-\frac{p+q}{2}=\frac{q-p}{2}=-\frac{p-q}{2}$ being multiplied by $F'(\frac{p+q}{2})$ cancel out.
Similarly, we can define the generalized {\em quasi-arithmetic Jensen-Bregman divergences} as:
\begin{equation}
\JB_F^{\rho,\tau}(p,q)=\frac{B_F^{\rho,\tau}\left(p:M_\rho(p,q)\right)+B_F^{\rho,\tau}\left(q:M_\rho(p,q)\right)}{2}.
\end{equation}

Consider $\tau=\idf$. Since $\rho(M_\rho(p,q))=\frac{\rho(p)+\rho(q)}{2}$, and $\rho(p)-\rho(M_\rho(p,q))=\frac{\rho(p)-\rho(q)}{2}  =-(\rho(q)-\rho(M_\rho(p,q))$ we get the following identity:
\begin{equation}
\JB_F^{\rho,\idf}(p,q)= \frac{F(p)+F(q)}{2} - F(M_\rho(p,q))  = J_F^{\rho,\idf}(p,q).
\end{equation}


\begin{lemma}[Generalized equivalence of Jensen-Bregman divergences with Jensen divergences]
The  $(M_\rho,M_\tau)$-Jensen-Bregman divergence amounts to a  $(M_\rho,M_\tau)$-Jensen divergence when $\tau=\idf$ ({\it ie.}, $M_\tau=A$, the arithmetic barycentric mean): $\JB_F^{\rho,\idf}(p,q)= J_F^{\rho,\idf}(p,q)$.
\end{lemma}

\subsection{Case study: Pythagorean-convex Bregman divergences}

Let us report the Bregman divergence with respect to a multiplicatively convex function:
\begin{example}
For the geometric mean $\rho(x)=\tau(x)=\log x$, we get the following geometric Bregman divergence ($(G,G)$-Bregman divergence or multiplicative Bregman divergence) for a $(G,G)$-convex generator function $F$:

\begin{equation}
B_F^G(p:q) =\lim_{\alpha\rightarrow 0}  \frac{1}{\alpha} J_{\tau,\alpha}(p:q) = F'(q)\KL(p:q)-\KL(F(p):F(q)),
\end{equation}
where $\KL(p:q)=p\log\frac{p}{q}$ is the renown Kullback-Leibler univariate function~\cite{BD-2005}.
\end{example}

\begin{corollary}[Pythagorean-convex Bregman divergences]
The Bregman divergences with respect to Pythagorean convexity are:
\begin{eqnarray}
B_F^{A,A}(p:q) &=& B_F(p:q) = F(p)-F(q)-(p-q)F'(q), \quad F\in\calC \\
B_F^{G,G}(p:q) &=& F(p)\log\frac{F(q)}{F(p)} + \left(p\log\frac{p}{q}\right)F'(q), \quad F\in\calC_{G,G}\\
B_F^{H,H}(p:q) &=& F^2(q) \left(  \frac{1}{F(q)} -\frac{1}{F(p)} \right)  + p^2  \left(  \frac{1}{p} -\frac{1}{q} \right)F'(q), \quad F\in\calC_{H,H} \\
\end{eqnarray}
\end{corollary}
Similarly, the six other Pythagorean-convexity Bregman divergences can be uncovered. 

Let us introduce a notion of dominance between means as follows:
\begin{definition}[Dominance]
A mean $M$ is said to dominate a mean $N$ iff $M_\alpha(x,y)\geq N_\alpha(x,y)$ for all $x,y\in I$ and $\alpha\in [0,1]$.
\end{definition}

We write $M\geq N$ or $N\leq M$ when $M$ dominates $N$.

\begin{definition}[Comparable means]
Two  means $M$ and $N$ are said comparable if either $M$ dominates $N$ ({\em ie.}, $M\geq N$) or $N$ dominates $M$ ({\em ie.}, $M\leq N$).
\end{definition}

The power means are comparable means, and $P_{\delta_1}\leq P_{\delta_2}$ when $\delta_1<\delta_2$.
This explains the fundamental AGH inequality: $A=P_1\geq G=P_0 \geq H=P_{-1}$.
Lehmer means are comparable: $L_\delta\leq L_{\delta'}, \forall \delta\leq \delta'$.
The dual mean operator reverses the comparison order: If $M_1\leq M_2$ then $M_2^*\leq M_1^*$.

From the dominance relationships of means, we can get inequalities between these generalized  Jensen/Bregman  divergences.
For example, for an {\em increasing} function $F(x)$ that is both $(M_1,N_1)$-convex and $(M_2,N_2)$-convex, we have 
$J_F^{M_1,N_1}(p:q)\geq J_F^{M_2,N_2}(p:q)$ when $N_1\geq N_2$ and  $M_1 \leq M_2$.
Indeed, we check that in order to have:
\begin{equation}
N_1(F(p),F(q))-F(M_1(p,q)) \geq N_2(F(p),F(q))-F(M_2(p,q)),
\end{equation}
it is sufficient to have $N_1\geq N_2$ and $M_1 \leq M_2$ for an increasing function $F\in \calC_{N_1,M_1}\cap \calC_{N_2,M_2}$.

Let $\rho,\tau: I\rightarrow (0,\infty)$ be two synchronous (meaning both increasing or both decreasing) continuous bijective functions with $\tau/\rho$ nonincreasing.
Then  the quasi-arithmetic mean $M_\rho$ is dominated by a quasi-arithmetic mean $M_\tau$: $M_\rho\leq M_\tau$.

%
%

\subsection{Checking the quasi-arithmetic convexity of functions}

To check whether a function $F$ is $(M,N)$-convex or not when using quasi-arithmetic means, we can use a reduction to standard convexity as follows:

\begin{lemma}[$(M_\rho,M_\tau)$-convexity to  ordinary convexity~\cite{Aczel-1947}]\label{lemma:cvx}
Let $\rho:I\rightarrow\bbR$ and $\tau:J\rightarrow\bbR$ be two continuous and strictly monotonic real-valued functions with $\tau$ increasing,
then  function $F:I\rightarrow J$ is $(M_\rho,M_\tau)$-convex iff function $G=F_{\rho,\tau} = \tau\circ F\circ\rho^{-1}$ is (ordinary) convex on $\rho(I)$.
\end{lemma}

\begin{proof}
%
Let us rewrite the $(M_\rho,M_\tau)$-convexity  midpoint inequality as follows:

\begin{eqnarray}
F(M_\rho(x,y)) &\leq& M_\tau(F(x),F(y)),\\
F\left(\rho^{-1}\left(\frac{\rho(x)+\rho(y)}{2}\right)\right) &\leq& \tau^{-1}\left(\frac{\tau(F(x))+\tau(F(y))}{2}\right),\\
\end{eqnarray}

Since $\tau$ is strictly increasing, we have:

\begin{equation}
(\tau\circ F\circ\rho^{-1}) \left(\frac{\rho(x)+\rho(y)}{2}\right) \leq \frac{(\tau\circ F)(x)+(\tau\circ F)(y)}{2}.
\end{equation}

Let $u=\rho(x)$ and $v=\rho(y)$ so that $x=\rho^{-1}(u)$ and $y=\rho^{-1}(v)$ (with $u,v\in\rho(I)$).
Then it comes that:
\begin{equation}
(\tau\circ F\circ\rho^{-1})\left(\frac{u+v}{2}\right) \leq \frac{(\tau\circ F\circ\rho^{-1})(u)+(\tau\circ F\circ\rho^{-1})(v)}{2}.
\end{equation}
This last inequality is precisely the ordinary midpoint convexity inequality for  function $G=F_{\rho,\tau}=\tau\circ F\circ\rho^{-1}$.
Thus $(M_\rho,M_\tau)$-convex is convex iff $G$ is ordinary convex.
\end{proof}

For example, a function $F$ is $M_p$-convex (a shortcut for $(A,M_p)$-convex) iff $F^p$ is convex when $p>0$.
Moreover, every $M_p$ convex function belongs to the class of $M_q$ convex functions when $q\geq p$.

When the functions are twice differentiable, this lemma allows one to check whether a function is $(M_\rho,M_\tau)$ by checking whether $(\tau\circ F\circ\rho^{-1})''>0$ or not.
For example, a function $F$ is $(H,H)$-convex iff $\frac{1}{F(1/x)}$ is convex.
Recall that the $(H,H)$-convexity of $F$ implies the following generalized Jensen midpoint inequality:
\begin{equation}
\frac{2F(p)F(q)}{F(p)+F(q)} \geq F\left( \frac{2pq}{p+q} \right).
\end{equation}

Another example is to check the $(G,A)$-strict convexity of twice-differentiable $F$ by checking that $x^2F''(x)+xF'(x)>0$ for $x>0$, etc.

Notice that we can also graphically check the $(M_\rho,M_\tau)$-convexity of a univariate function $F$ by plotting the function
 $y=F_\tau(x_\rho)=(\tau \circ F)(x_\rho)$ with  abscissa $x_\rho=\rho^{-1}(x)$.

We can thus give a complete $(P_{\delta_1},P_{\delta_2})$-convex characterization of functions $f:I\subset \bbR_{++}\rightarrow \bbR_{++}$, see~\cite{Maksa-2015}.
Define function $f_{\delta_1,\delta_2}:I_{\delta_1}\rightarrow\bbR$ with $I_{\delta}=\{ x^\delta\ : \ x\in I\}$ for $\delta\not=0$ (and $I_0=\{ \log x\ : \ x\in I\}$) as:

\begin{equation}
f_{\delta_1,\delta_2}(x) = \left\{
\begin{array}{ll}
\sign(\delta_2)f^{\delta_2}(x^{\frac{1}{\delta_1}}) & \delta_1\not=0,\delta_2\not =0\\
\sign(\delta_2)(f^{\delta_2}(\exp(x))) & \delta_1=0,\delta_2\not =0\\
\log(f(x^{\frac{1}{\delta_1}})) &   \delta_1\not=0,\delta_2 =0\\
\log(f(\exp(x))) & \delta_1=0,\delta_2=0
\end{array}
\right.
\end{equation}

Then $f$ is  $(P_{\delta_1},P_{\delta_2})$-convex on $I\subset  \bbR_{++}$ iff $f_{\delta_1,\delta_2}$ is convex on $I_{\delta_1}$.

\subsection{Proper quasi-arithmetic Bregman divergences}
 
Applying the ordinary Bregman divergence on the ordinary convex generator $G(x)=\tau (F(\rho^{-1}(x)))$ for a $(M_\rho,M_\tau)$-convex function 
with:
\begin{eqnarray}
G'(x) &=& \tau(F(\rho^{-1}(x)))'\\
&=& (F(\rho^{-1}(x)))' \tau'(F(\rho^{-1}(x)))\\
&=& (\rho^{-1}(x))' F'(\rho^{-1}(x)) \tau'(F(\rho^{-1}(x)))\\
&=& \frac{1}{(\rho'(\rho^{-1})(x))}  F'(\rho^{-1}(x)) \tau'(F(\rho^{-1}(x))),
\end{eqnarray}
we get an ordinary Bregman divergence that is, in general, {\em different} from the generalized quasi-arithmetic Bregman divergence $B^{\rho,\tau}_F$:

\begin{eqnarray}
B_G(p:q) &=& G(p)-G(q)-(p-q)G'(q),\\
B_G(p:q) &=& \tau(F(\rho^{-1}(p)))-\tau(F(\rho^{-1}(q)))-(p-q) \frac{1}{(\rho'(\rho^{-1})(q))}  F'(\rho^{-1}(q)) \tau'(F(\rho^{-1}(q)))    
\end{eqnarray}

This is in general a different generalized Bregman divergence: $B_G(p:q)\not = B_F^{\rho,\tau}(p:q)$.
But we check that $B_G(p:q) = B_F^{\rho,\tau}(p:q)$ when $\rho(x)=\tau(x)=x$ (since we have the derivatives $\rho'(x)=\tau'(x)=1$).


Let us notice the following remarkable identity:

\begin{equation}\label{eq:rkid}
\boxed{B_F^{\rho,\tau}(p:q)=\frac{1}{\tau'(F(q))} B_G(\rho(p):\rho(q))}
\end{equation}

Since $\tau(x)$ is a strictly increasing function, we have $\tau'(x)>0$, and since $B_G$ is an ordinary Bregman divergence we have $B_G(p':q')\geq 0$
 (and $B_G(p':q')=0$ iff $p'=q'$) for any pair of $p'$ and $q'$ of values. It follows that  $B_F^{\rho,\tau}$ is a {\em proper generalized Bregman divergence}:
$B_F^{\rho,\tau}(p:q)\geq 0$ with equality iff $p=q$.

Notice that $B_G(\rho(p):\rho(q))=B_{H}(p:q)$ with $H=\tau\circ F$.
However function $H$ may not be strictly ordinary convex 
 since $H'(x)=F'(x)\tau'(F(x))$, and $H''(x)=F''(x)\tau'(F(x))+(F'(x))^2\tau''(F(x))$, and therefore $B_H$ may not be a Bregman divergence.
Function $H$ is strictly convex when $H''(x)=F''(x)\tau'(F(x))+(F'(x))^2\tau''(F(x))>0$.

\begin{theorem}[Proper generalized $(M_\rho,M_\tau)$-Quasi-Arithmetic Bregman divergence]
Let $F:I\subset \bbR \rightarrow \bbR$ be a real-valued $(M_\rho,M_\tau)$-convex function defined on an interval $I$ for two strictly monotone and differentiable functions $\rho$ and $\tau$, with $\tau$ strictly increasing.
The quasi-arithmetic Bregman divergence induced by the comparative convexity is a {\em proper} divergence:
\begin{eqnarray}
B^{\rho,\tau}_F(p:q) 
&=&  
\frac{\tau(F(p))-\tau(F(q))}{\tau'(F(q))} - \frac{\rho(p)-\rho(q)}{\rho'(q)} F'(q),\\
&=& \frac{1}{\tau'(F(q))} B_{\tau\circ F\circ \rho^{-1}}(\rho(p):\rho(q)) \geq 0,
\end{eqnarray}
with $B^{\rho,\tau}_F(p:q)=0$ iff $p=q$.
\end{theorem}

Using Taylor's first-order expansion with the exact Lagrangre remainder, we get:

$$
B^{\rho,\tau}_F(p:q) = \frac{1}{\tau'(F(q))} (\rho(p)-\rho(q))^2 G''(\rho(\xi)),
$$
for $\xi\in [pq]$.

\subsection{Conformal Bregman divergences in embedded space}
The generalized Bregman divergences $B_F^{\rho,\tau}(p:q)$ can also be interpreted as Bregman conformal divergences~\cite{Conformal-2016} on the $\rho$-embedding of parameters:
$B_F^{\rho,\tau}(p:q)= \kappa(\rho(q))  B_G(\rho(p):\rho(q))$ 
with positive conformal factor $\kappa(x)=\frac{1}{\tau'(F(\rho^{-1}(x)))}$ for a strictly increasing monotone function $\tau$.
 
\begin{corollary}[Generalized quasi-arithmetic Bregman divergences as conformal Bregman divergences]
The generalized quasi-arithmetic Bregman divergence $B^{\rho,\tau}_F(p:q)$ amounts to compute an ordinary  Bregman conformal divergence in the  $\rho$-embedded space:
$B_F^{\rho,\tau}(p:q)= \kappa(\rho(q))  B_G(\rho(p):\rho(q))$ with  conformal factor $\kappa(x)=\frac{1}{\tau'(F(\rho^{-1}(x)))}>0$.
\end{corollary}

\subsection{Generalized Bregman centroids}

The identity of Eq.~\ref{eq:rkid} allows one to compute generalized Bregman centroids easily.
For a positively weighted set of $n$ scalars $(w_1,p_1), \ldots, (w_n,p_n)$, 
define the generalized Bregman centroid as the minimiser of $\sum_{i=1}^n w_i B_F^{\rho,\tau}(c:p_i)$.
We have:
$$
\sum_{i=1}^n w_i B_F^{\rho,\tau}(c:p_i)=\sum_{i=1}^n \frac{w_i}{\tau'(F(p_i))} B_G(\rho(c):\rho(p_i)).
$$

Let $w_i'=\frac{w_i}{\tau'(F(p_i))}>0$, $W'=\sum_{i=1}^n w_i'$, $c'=\rho(c)$ and $p_i'=\rho(p_i)$.
Then it follows that the generalized Bregman centroid is unique and available in closed-form:

\begin{equation}
G'(c')= \sum_{i=1}^n \frac{w_i'}{W'} G'(p_i'),
\end{equation}
with $G'(x')=G'(\rho(x))=\frac{F'(x)\tau'(F(x))}{\rho'(x)}$.
Thus the generalized Bregman centroid can be interpreted as a regularized Bregman centroid (see the total Bregman centroid~\cite{tBD-2011}).
We can extend the $k$-means++ seeding~\cite{kmeanspp-2007,kvariatepp-2016}.

\subsection{Examples of quasi-arithmetic Bregman divergences}

Using Table~\ref{tab:kappa} and Eq~\ref{eq:BDkappa}, we can easily instantiate the various comparative-convexity divergences.
For example, considering $\rho(x)=x$ (arithmetic mean $A$), we get the following families of divergences:

\begin{description}

\item[$(A,A)$-divergences.] Ordinary case with $\rho=\tau=\mathrm{id}$, the identity function.

\begin{eqnarray}
J_F^{A,A}(p;q) &=& \frac{F(p)+F(q)}{2}-F\left(\frac{p+q}{2}\right),\\
J_{F,\alpha}^{A,A}(p:q) &=&  (1-\alpha)F(p)+\alpha F(q)- F( (1-\alpha)p+\alpha q),\\
B_F^{A,A}(p:q) &=& F(p)-F(q)-(p-q)F'(q). 
\end{eqnarray}

$F\in\calC_{A,A}$ iff $F\in\calC$.

\item[$(A,G)$-divergences.] $\rho=\mathrm{id}$, $\tau=\log$ the logarithmic function.

\begin{eqnarray}
J_F^{A,G}(p;q) &=& \sqrt{F(p)F(q)}-F\left(\frac{p+q}{2}\right),\\
J_{F,\alpha}^{A,G}(p:q) &=&  F^{1-\alpha}(p)F^\alpha(q)- F( (1-\alpha)p+\alpha q),\\
B_F^{A,G}(p:q) &=& F(q)\log \frac{F(p)}{F(q)} -(p-q)F'(q). 
\end{eqnarray}

$F\in\calC_{A,G}$ iff $\log\circ F$ is convex ({\it ie.}, a log-convex function).

\item[$(A,H)$-divergences.] $\rho=\mathrm{id}$, $\tau=\frac{1}{x}$ (with $\tau'(x)=-\frac{1}{x^2}$).

\begin{eqnarray}
J_F^{A,H}(p;q) &=& \frac{2F(p)F(q)}{F(p)+F(q)} -F\left(\frac{p+q}{2}\right),\\
J_{F,\alpha}^{A,H}(p:q) &=&  \frac{1}{(1-\alpha)\frac{1}{F(p)}+ \alpha\frac{1}{F(q)} }- F( (1-\alpha)p+\alpha q),\\
  &=&  \frac{F(p)F(q)}{\alpha F(p)+(1-\alpha)F(q)}- F( (1-\alpha)p+\alpha q),\\
B_F^{A,H}(p:q) &=&  F^2(q)\left(\frac{1}{F(q)} - \frac{1}{F(p)} \right)-(p-q)F'(q). 
\end{eqnarray}

$F\in\calC_{A,H}$ iff $\frac{1}{x} \circ F=\frac{1}{F}$ is convex.
For example, $F(x)=\frac{1}{x\log x}$ is $(A,H)$-convex on $x>0$.

\item[$(A,M_\delta)$-divergences.] $\rho=\mathrm{id}$, $\tau=x^\delta$ for $\delta >0$.

\begin{eqnarray}
J_F^{A,M}(p;q) &=& \left(\frac{F^\delta(p)+F^\delta(q)}{2}\right)^{\frac{1}{\delta}} -F\left(\frac{p+q}{2}\right),\\
J_{F,\alpha}^{A,M}(p:q) &=& \left((1-\alpha)F^\delta(p)+\alpha F^\delta(q)\right)^{\frac{1}{\delta}}  - F( (1-\alpha)p+\alpha q),\\
B_F^{A,M}(p:q) &=&   \left(\frac{F(q)^\delta-F(p)^\delta}{\delta F(p)^{\delta-1}}   \right)-(p-q)F'(q). 
\end{eqnarray}

$F\in\calC_{A,M_\delta}$ iff $(x^\delta) \circ F=F^\delta$ is convex for $\delta>0$.

\end{description}

\section{Generalized statistical Bhattacharyya distances with comparative means}\label{sec:GB}

The Bhattacharyya distance~\cite{Bhattachayya-1943} (1943) is a {\em statistical distance} defined between two probability measures dominated by a measure $\nu$ (often, the Lebesgue measure or the counting measure).
Let $p(x)$ and $q(x)$ be the densities defined on the support $\calX$.
Then the  Bhattacharyya distance is defined  by:
\begin{equation}
\Bhat(p(x):q(x)) =  -\log \int_\calX \sqrt{p(x)q(x)} \dnu(x).
\end{equation}
The skewed Bhattacharyya distance for $\alpha\in(0,1)$ is defined by
\begin{equation}
\Bhat_\alpha(p(x):q(x)) =  -\log \int_\calX p^{\alpha}(x) q^{1-\alpha}(x) \dnu(x),
\end{equation}
with $\Bhat(p(x):q(x))=\Bhat_{\frac{1}{2}}(p(x):q(x))$.

The term $c_\alpha(p(x):q(x))=\int_\calX p^{\alpha}(x) q^{1-\alpha}(x) \dnu(x)$ is interpreted as a coefficient of similarity, also called the  {\em Bhattacharyya affinity} coefficient~\cite{GenBhat-2014}.
This term plays an important role in information geometry~\cite{IG-2016} within the family of $\alpha$-divergences:
\begin{equation}
I_\alpha(p(x):q(x)) =   \frac{1-\int_\calX p^{\alpha}(x) q^{1-\alpha}(x) \dnu(x)}{\alpha(1-\alpha)} = \frac{1-c_\alpha(p(x):q(x))}{\alpha(1-\alpha)}.
\end{equation}
Thus we can plug the Bhattacharyya distance in the $\alpha$-divergence by using the identity $c_\alpha(p(x):q(x))=\exp(-\Bhat_\alpha(p(x):q(x)))$.
The $\alpha$-divergences tend to the Kullback-Leibler divergence when $\alpha\rightarrow 1$ and to the reverse Kullback-Leibler divergence when  $\alpha\rightarrow 0$, see~\cite{IG-2016}.

The standard Bhattacharyya distance is very well suited to the computation of the distance between members of the same exponential family~\cite{BR-2011}.
Indeed, let $p(x)=p(x;\theta_p)$ and $q(x)=p(x;\theta_q)$ be two distributions belonging to the same exponential family
$\{p(x;\theta) =\exp(\theta^\top x-F(\theta)) \ :\ \theta\in\Theta \}$, where $\Theta$ denotes the natural parameter space~\cite{EF-2009}.
Then we have:
\begin{equation}
\Bhat_\alpha(p(x;\theta_p):p(x;\theta_q))  = J_{F,1-\alpha}(\theta_p:\theta_q).
\end{equation}
Here, the term $1-\alpha$ comes from the fact that the coefficient has been historically defined for the geometric mean $p^\alpha(x) q^{1-\alpha}(x)$.

In~\cite{HolderDiv-2017}, the  Bhattacharyya distance was extended to positive measures by defining a projective divergence relying on the H\"older inequality.
By definition, any {\em projective divergence} $D(p,q)$ satisfies $D(\lambda p,\lambda' q)=D(p,q)$ for any $\lambda,\lambda'>0$.
Here, we consider yet another rich generalization of the Bhattacharyya distance  by noticing that $p^{\alpha}(x) q^{1-\alpha}(x)=G(q(x),p(x);\alpha,1-\alpha)$ is the geometric barycenter and that
$\int_\calX ({\alpha} p(x)+ {1-\alpha} q(x)) \dnu(x)=\int_\calX  A(p(x),q(x);1-\alpha,\alpha) \dnu(x)= 1$ can be interpreted as a (hidden) unit denominator.

Thus consider two {\em comparable means}~\cite{Cargo-1965} $M$ and $N$ that guarantees by definition that $M(a,b;\alpha,1-\alpha) \leq N(a,b;\alpha,1-\alpha)$ for any value of $a,b$ and $\alpha\in [0,1]$ (written for short as $M\leq N$), and 
define the generalized  Bhattacharyya distance as follows:

\begin{definition}[Comparative-Mean Bhattacharyya Distance, CMBD]\label{def:GB}
For two distinct comparable means $M$ and $N$ such that $M \leq N$, the comparative-mean skewed Bhattacharyya distance is defined by:
\begin{equation}\label{eq:GenBhat}
\boxed{
\Bhat_\alpha^{M,N}(p(x):q(x)) = -\log \frac{\int_\calX M(p(x),q(x);1-\alpha,\alpha) \dnu(x)}{\int_\calX N(p(x),q(x);1-\alpha,\alpha) \dnu(x)}.
}
\end{equation}
\end{definition}
We have $\Bhat_\alpha^{M,N}(q(x):p(x)) = \Bhat_{1-\alpha}^{M,N}(p(x):q(x))$.
It follows from the property of abstract barycentric means that $\Bhat_\alpha^{M,N}(q(x):p(x))=0$ iff $M(p(x),q(x);1-\alpha,\alpha)=N(p(x),q(x);1-\alpha,\alpha)$ for strict distinct means, that is iff $p(x)=q(x)$.

When $M=G$ is chosen as the geometric mean and $N=A$ is taken as the arithmetic mean, we recover the ordinary skewed Bhattacharyya distance, modulo the fact that we swap $\alpha\leftrightarrow 1-\alpha$: $\Bhat_\alpha^{M,N}(p(x):q(x)) = \Bhat_{1-\alpha}(p(x):q(x))$.

When $M$ and $N$ are both {\em homogeneous} means, we end up with a {\em homogeneous} comparative-Mean Bhattacharyya distance.
That is, the divergence is invariant for the {\em same} scaling factor $\lambda$:
$\Bhat_\alpha^{M,N}(\lambda p(x): \lambda q(x)) = \Bhat_\alpha^{M,N}(p(x):q(x))$ for any $\lambda>0$.
See~\cite{Zhang-2013} for the definition of the homogeneous $(\alpha,\beta)$-divergences.

\begin{corollary}
The comparative-Mean Bhattacharyya distance for comparable homogeneous means yields a homogeneous statistical distance.
\end{corollary}


Since distinct power means are always comparable (ie., $P_{\delta_1} \leq P_{\delta_2}$ when $\delta_1<\delta_2$), we define the 
{\em power-mean Bhattacharyya divergence} for $\delta_1, \delta_2\in\bbR\backslash\{0\}$ with $\delta_1\not=\delta_2$ as follows:

\begin{eqnarray}\label{eq:powermeanbhat}
\Bhat_\alpha^{\delta_1,\delta_2}(p(x):q(x)) &=&  \frac{1}{\delta_1-\delta_2} \log\left( \frac{\int_\calX P_{\delta_1}(p(x),q(x);1-\alpha,\alpha) \dnu(x)}{\int_\calX P_{\delta_2}(p(x),q(x);1-\alpha,\alpha) \dnu(x)}\right),\\
&=&  \frac{1}{\delta_1-\delta_2} \log  \left( \frac{\int_\calX \left((1-\alpha)p^{\delta_1}(x)+\alpha q^{\delta_1}(x)\right)^{\frac{1}{\delta_1}} \dnu(x)}{
\int_\calX \left((1-\alpha)p^{\delta_2}(x)+\alpha q^{\delta_2}(x)\right)^{\frac{1}{\delta_2}} \dnu(x)
}\right).
\end{eqnarray}

\begin{remark}
Yet another type of divergences are {\em conformal divergences}~\cite{Conformal-2016}.
A conformal divergence $D_{h}(x:y)$ can be factorized as $D_h(x:y)=h(x:y)D(x:y)$ where $h(x:y)$ is a positive {\em conformal factor} 
function~\cite{tJ-2015} and $D$ a base divergence.
Conformal divergences such as the total Jensen divergences~\cite{tJ-2015} or the total Bregman divergences~\cite{tB-2012} proved useful in practice to regularize the base divergence and to guarantee invariance by rotation of the coordinate system.
\end{remark}

When considering quasi-arithmetic means  for $M=M_f$ and $N=M_g$ for two continuous and increasing functions $f$ and $g$ on an interval domain $I=[a,b]$, a necessary and sufficient condition~\cite{Cargo-1965}
 for $M_f\leq M_g$ is that $g\circ f^{-1}$ is convex on  interval $[f(a),f(b)]$.
Function $g$ is then said convex with respect to $f$.
Thus two quasi-arithmetic means $M_f$ and $M_g$ are comparable when either  $g\circ f^{-1}$ is convex ($M_f\leq M_g$) or 
$f\circ g^{-1}$ is convex ($M_f\geq M_g$).

{\em Relative convexity} studies the concept of convexity of a function $g$ with respect to
another function $f$: It is denoted by $f \triangleleft g$, with the notation $\triangleleft$ borrowed from~\cite{ConvexFunction-2006}. 
 A general characterization of the relative convexity $f\triangleleft g$ is as follows:

\begin{equation}
\forall x,y,z\in\calX, f(x)\leq f(y)\leq f(z) \Rightarrow \left| \begin{array}{ccc}
1 & f(x) & g(x)\cr
1 & f(y) & g(y)\cr
1 & f(z) & g(z)\cr
\end{array} \right|\geq 0,
\end{equation}
for $f$ a non-constant function.

When the domain $\calX=I$ is an interval and $f$ is a strictly increasing and continuous function, then 
$M_f\leq M_g$ iff $f\triangleleft g$ ($g\circ f^{-1}$ is convex).

For example, $F$ is multiplicatively convex (type $(G,G)$-convexity) iff:
\begin{equation}
\forall x,y,z\in\calX, x \leq y \leq z,\quad
\Rightarrow \left| \begin{array}{ccc}
1 & \log x & f(x)\cr
1 & \log y & f(y)\cr
1 & \log z & f(z)\cr
\end{array} \right|\geq 0,
\end{equation}

Relative convexity is a sub-area of comparative
convexity.
For example, we have the following  correspondences of comparative convexity classes of functions:
\begin{itemize}
	\item $f\in\calC$  iff $\idf \triangleleft f$,
	\item $f\in\calC_{A,G}$ iff $\idf \triangleleft \log f$,
	\item $f\in\calC_{G,A}$ iff $\log \triangleleft  f$,
	\item $f\in\calC_{G,G}$ iff $\log \triangleleft \log f$.
\end{itemize}

The criterion of relative convexity can be used to recover the celebrated Arithmetic Mean-Geometric Mean-Harmonic Mean (AM-GM-HM inequality):
When $M=M_{\log}$ is chosen as the geometric mean and $N=M_{\idf}$ is taken as the arithmetic mean, 
we check that we have $\log \triangleleft \idf$ (and $\idf\circ \exp=\exp$, $M_{\log} \leq M_{\idf}$), and we recover the ordinary skewed Bhattacharyya distance. 

An interesting subfamily of Bhattacharyya distance is obtained for $N=A=M_\idf$. In that case, we have:

\begin{equation}\label{eq:GenBhatOne2}
\boxed{
\Bhat_\alpha^{M,N}(p(x):q(x)) = -\log  \int_\calX M(p(x),q(x);1-\alpha,\alpha) \dnu(x),
}
\end{equation}
for $M_f\leq M_\idf$ with $f^{-1}$ ordinary convex.

We compare the $\delta$-th H\"older power mean with the $\delta'$-th H\"older power mean on $\bbR_{++}=(0,\infty)$ as follows:
$P_\delta\leq P_{\delta'}$ when $\delta\leq\delta'$ and $P_\delta\geq P_{\delta'}$ when $\delta\geq\delta'$.

Notice that in~\cite{GenBhat-2014}, the generalized Bhattacharyya coefficients $c_\alpha^f(p(x):q(x))=\int_\calX M_f(p(x),q(x);1-\alpha,\alpha) \dnu(x)$ were introduced to upper bound the Bayes' probability of error.
Here, we further extend this generalization by considering comparable means, and we rewrite the comparative-mean Bhattacharyya distance as:
\begin{equation}\label{eq:GenBhat2}
\boxed{\Bhat_\alpha^{M,N}(p(x):q(x)) = -\log \frac{ c_\alpha^M(p(x):q(x)) }{c_\alpha^N(p(x):q(x))}}
\end{equation}
where $c_\alpha^M(p(x):q(x))=\int_\calX M(p(x),q(x),\alpha,1-\alpha) \dnu(x)$ is a generalized Bhattacharyya affinity coefficient.
Notice that $c_\alpha^A(p(x):q(x))=1$ when choosing the arithmetic mean.


Those generalized  Bhattacharyya distances are handy for getting closed-form formulas depending on the structure of the probability densities:
For example, consider the harmonic-arithmetic comparative-mean Bhattacharyya distance $\Bhat_\alpha^{H,A}(p(x):q(x))$ between two Cauchy distributions
$p(x)=p(x;s_1)$ and $q(x)=p(x;s_2)$ with $p(x;s)=\frac{s}{\pi(x^2+s^2)}$ for a scale parameter $s>0$.

It was shown in~\cite{GenBhat-2014} that:
\begin{equation} 
c_\alpha^H(p(x;s_1):p(x;s_2))= \frac{s_1s_2}{((1-\alpha)s_1+\alpha s_2 )s_\alpha}.
\end{equation}

Therefore it comes that:
\begin{equation} 
\Bhat_\alpha^{H,A}(p(x;s_1):p(x;s_2)) = -\log \frac{s_1s_2}{((1-\alpha)s_1+\alpha s_2 )s_\alpha}.
\end{equation}

The original $(G,A)$-Bhattacharyya distance does not allow to get a simple closed-form expression when dealing with Cauchy distributions.
It is thus a mathematical trick to tailor the generalized Bhattacharyya distance  to the structure of the family of distributions in practice to get efficient algorithms.

Note that the Cauchy family of distributions do not form an exponential family, but can be interpreted as a deformed exponential family~\cite{Naudts-2014} by defining corresponding deformed logarithm and exponential functions.

Other examples of generalized Bhattacharyya coefficients with closed-form expressions are reported in~\cite{GenBhat-2014} using the power means for the Pearson type VII and multivariate $t$-distributions.

Notice that in the discrete case, we get a closed-form expression since integrals transforms into finite sums:

\begin{equation}
\Bhat_\alpha^{M,N}(p(x):q(x)) = -\log \frac{\sum_{i=1}^d M(p_i,q_i;1-\alpha,\alpha)}{\sum_{i=1}^d  N(p_i,q_i;1-\alpha,\alpha)}.
\end{equation}

There are many statistical distances available that prove useful depending on application context~\cite{Basseville-2013}.
Comparative means allow one to define yet another one as follows:

\begin{remark}
Note that comparable  means~\cite{Cargo-1965} satisfying  $M_f\leq M_g$ allows to define a symmetric distance gap:
\begin{equation}
D_{f,g}(p(x),q(x)) = \int \left(M_g(p(x),q(x))-M_f(p(x),q(x))\right)\dnu(x)\geq 0
\end{equation}
A necessary and sufficient condition for $f,g:I=[a,b]\rightarrow\bbR$ is to have $g\circ f^{-1}$ convex on interval $[f(a),f(b)]$, see~\cite{Cargo-1965}.
\end{remark}

\section{Conclusion and discussion}\label{sec:concl}

We defined generalized Jensen divergences (Definition~\ref{def:jccd}) and generalized Bregman divergences (Definition~\ref{def:bccd}) using the framework of comparative convexity based on abstract means.
In particular, we reported a closed-form formula for the generalized  Bregman divergences (Theorem~\ref{theo:qabd}) when considering quasi-arithmetic means.
We proved that those generalized quasi-arithmetic Bregman divergences are proper divergences that can be interpreted as conformal ordinary Bregman divergences on  an embedded representation of input space.
Those generalized Bregman divergences can be fruitfully used in machine learning: Not only can we learn~\cite{LearningDivergence-2015} the separable convex generators $F_i$'s component-wise, but we can also learn the increasing functions $\rho_i$ and $\tau_i$ that induces the quasi-arithmetic means $M_{\rho_i}$ and $M_{\tau_i}$.
Finally, we introduced a generalization of the Bhattacharyya statistical distance (Definition~\ref{def:GB}) and of the  Bhattacharyya coefficient for comparable means, and show that depending on the structure of the distributions we may obtain handy closed-form expressions or not.
In particular, the generalized Bhattacharyya distances yield homogeneous divergences when homogeneous comparative means are used.
Since minimizing skewed Bhattacharyya distances allows one to bound the probability of error, we may define similarly generalized Chernoff information~\cite{Chernoff-2013}, etc.

This work emphasizes that the theory of means are at the very heart of distances.
In general, a family of means can be investigated by studying comparison and equality, homogeneity, and exact characterization by functional equations or inequalities. There are many means~\cite{Bullen-2013} ({\it eg.}, Heinz means~\cite{Heinz-2006}, Gini means, Lehmer means, counter-harmonic means, Whitely means, Muirhead means)  to consider to build and study novel family of distances and statistical distances. To derive generalized Bregman divergences, we need to study the asymptotic expansion of barycentric means. Generalization of means to weighted means is an interesting topic in itself~\cite{Witkowski-2006}.
We may also consider asymmetric weighted means~\cite{Qi-2000} to define corresponding Bregman divergences. 
Properties of a family of means may also be considered.
For example, the power means form a {\em scale}~\cite{Pasteczka-2015}: 
That means that there exists a bijection between $\delta\in\bbR$ and $P_\delta(u,v)$ for $u\not =v$.
Informally speaking, the family of power means allows one to interpolate between the minimum and the maximum.
Although the power means are the only homogeneous quasi-arithmetic means that form a scale, there exist other families of quasi-arithmetic means
that form a scale~\cite{Pasteczka-2015}.  Means can also be defined for various types of data like matrices~\cite{Petz-2005}.

As a final remark, let us notice that we have defined several divergences using (extrinsic) means. 
But divergences $D(\cdot:\cdot)$ induce a geometry where we can also be used to define (intrinsic) means $m$ (commonly called centers) by minimizing the loss function $\frac{1}{n} \sum_{i=1}^n D(p_i:m)$.

\section*{Acknowledgments}
The authors would like to thank Gautier Marti and Ga\"etan Hadjeres for reading a preliminary draft.


\appendix

\section{Axiomatization of the family of quasi-arithmetic means\label{sec:aqam}}

Consider the following axiomatization of a  mean:

\begin{description}

\item[QAM1 (Reflexivity).] The mean of two identical elements is that element: $M(x,\ldots,x)=x$,

\item[QAM2 (Symmetry).] The mean is a symmetric function: $M(\sigma(x_1),\ldots,\sigma(x_n))=M(x_1,\ldots,x_n)$ for any permutation $\sigma$,

\item[QAM3 (Continuity and monotonicity).] The mean function $M(\cdot)$ is continuous and increasing in each variable, 

\item[QAM4 (Associativity).] The mean $m=M(x_1,\ldots x_n)$ is invariant when some elements are replaced by the partial mean value $m_{i-1}=M(x_1,\ldots,x_{i-1})$:  
$M(\underbrace{m_{i-1},\ldots,m_{i-1}}_{i-1},x_i,\ldots x_n)=m$ for all $i\in \{2,\ldots, n\}$. 

\end{description}

Kolmogorov~\cite{Kolmogorov-1930} proved in 1930 that the quasi-arithmetic mean $m$ writes for a continuous monotone function $f$ as:

\begin{equation}\label{eq:qamk}
m = M(x_1,\ldots,x_n) = f^{-1}\left(  \frac{1}{n} \sum_{i=1}^n f(x_i) \right).
\end{equation}

This quantity is called a quasi-arithmetic mean because we have $f(m)=\sum_{i=1}^n f(x_i)$: 
That is, the $f$-representation of the mean $m$ is the arithmetic mean of the $f$-representations of the elements.

There is a growing interest in generalizing quasi-arithmetic means.
In~\cite{GenQAM-2010}, the mean defined for two strictly continuous and monotonic functions $f$ and $g$ is defined by:
$$
M_{f,g}(p,q)=(f+g)^{-1}(f(p)+f(q)).
$$
The homogeneous means of that family coincide with the power (quasi-arithmetic) means~\cite{GenQAM-2010}.
This family of bivariate means can be extended to multivariate means: $M_{f_1,\ldots, f_n}(p,q)=(\sum_{i=1}^n f_i)^{-1}((\sum_{i=1}^n f_i(x_i))$, with functions $f_i$'s strictly continuous and monotonic of the same type.

Similarly, a {\em quasi-arithmetic expected value}~\cite{mean-2016} can be defined for a random variable $X\sim p(x)$ as follows:

\begin{eqnarray}
E_f[X] &=& f^{-1}\left( E_p[f(X)] \right),\\
&=& f^{-1}\left( \int_\calX p(x)f(x) \dnu(x) \right).
\end{eqnarray}

Notice that this definition of a quasi-arithmetic expected value coincides with the mean of a finite set of values of Eq.~\ref{eq:qamk} by taking the corresponding discrete distribution. Furthermore, the quasi-arithmetic expected value can be extended to positive and integrable densities $p(x)$ as follows: 
\begin{equation}
E_f[X] = f^{-1}\left( \frac{\int_\calX p(x)f(x)}{\int_\calX p(x)\dnu(x)} \dnu(x) \right).
\end{equation}

For example, $E_x[X]=E[X]$ is the expected value, $E^G[X]=E_{\log(x)}[X]=\exp(E[\log(X)])$ is the {\em geometric expected value}~\cite{Galton-1879}, 
and $E^H[X]=E_{\frac{1}{x}}[X]=\frac{1}{E[\frac{1}{X}]}$ is the {\em harmonic expected value}.

See also~\cite{Gibilisco-2017} for generalized Jensen inequalities of expectations of positive random variables with means.

\section{Lehmer Bregman divergences}

To illustrate the construction of generalized Bregman divergences beyond the quasi-arithmetic Bregman divergences, let us choose another family of parametric means: The Lehmer means.
Although the Lehmer means $L_\delta$ are not always regular, let us consider the subclass of regular Lehmer means.
The {\em weighted Lehmer mean}  $L_\delta$ for two positive reals $x$ and $y$ is defined by: 
\begin{equation}
L_\delta(x,y;\alpha)=\frac{(1-\alpha)x^{1+\delta}+\alpha y^{1+\delta}}{(1-\alpha)x^{\delta}+\alpha y^{\delta}}.
\end{equation}

Using the first-order Taylor expansion of $\frac{1}{1+x}=1-x+o(x^2)$, 
we get the following approximation of the Lehmer weighted mean when $\alpha\rightarrow 0$:
\begin{eqnarray}
L_{\delta,\alpha}(p,q) = L_\delta(p,q;\alpha) &=& \frac{p^{1+\delta}+\alpha(q^{1+\delta}-p^{1+\delta})}{p^{\delta}+\alpha(q^{\delta}-p^{\delta})},\\
&\simeq& \left(p+\alpha \frac{q^{1+\delta}-p^{1+\delta}}{p^\delta}\right) \left(1-\alpha \frac{q^\delta-p^\delta}{p^\delta}\right),\\
&\simeq& p+\alpha \frac{q^{1+\delta}-q^\delta-p^{1+\delta}+p^\delta}{p^\delta}.
\end{eqnarray}

Note that when $\delta=0$, the Lehmer mean $L_0$ equals the arithmetic mean $A$, and we find that $L_\delta\simeq p+\alpha(q-p)$, as expected.
When $\delta=-1$, the Lehmer weighted mean $L_{-1}$ equals the harmonic weighted mean $H(p,q;\alpha)=\frac{pq}{(1-\alpha)q+\alpha p}$, 
and we find that  $L_{-1}(p,q;\alpha)=H(p,q;\alpha)\simeq p+\alpha(1-\frac{p}{q})$.
The Taylor expansion for the  quasi-arithmetic mean (with generator $\tau(x)=\frac{1}{x}$) yielded $H(p,q;\alpha')\simeq p+\alpha'(p-\frac{p^2}{q})=p+(\alpha' p)(1-\frac{p}{q})$.
So by choosing $\alpha=\alpha' p$ (and $\alpha\rightarrow 0$ when  $\alpha'\rightarrow 0$), the two Taylor expansions obtained by the Lehmer and the quasi-arithmetic expressions match.

Let us define the Lehmer Bregman divergence with respect to two regular Lehmer means $L_\delta$ and $L_{\delta'}$ by:
\begin{equation}
B_F^{L(\delta,\delta')}(p:q) = \lim_{\alpha\rightarrow 0} \frac{1}{\alpha}\left(L_{\delta',\alpha}(F(p),F(q))  - F(L_{\delta,\alpha}(p,q)) \right).
\end{equation}

When $\alpha\rightarrow 0$, we have
\begin{equation}
F(L_{\delta,\alpha}(p,q)) \simeq F(p)+ \alpha \frac{q^{1+\delta}-q^\delta-p^{1+\delta}+p^\delta}{p^\delta} F'(p)
\end{equation}
and
\begin{equation}
L_{\delta',\alpha}(F(p),F(q)) \simeq F(p)+\alpha \frac{F(q)^{1+\delta'}-F(q)^{\delta'}-F(p)^{1+\delta'}+F(p)^\delta}{F(p)^{\delta'}}
\end{equation}

It follows that: 
\begin{equation}
B_F^{L(\delta,\delta')}(p:q)= 
 \frac{F(q)^{1+\delta'}-F(q)^{\delta'}-F(p)^{1+\delta'}+F(p)^{\delta'}}{F(p)^{\delta'}} - \frac{q^{1+\delta}-q^\delta-p^{1+\delta}+p^\delta}{p^\delta} F'(p)
\end{equation}

Let $\chi_\delta(p:q)=\frac{q^{1+\delta}-q^\delta-p^{1+\delta}+p^\delta}{p^\delta}$.
Then we get the compact expression of Lehmer Bregman divergences:

\begin{equation}
B_F^{L(\delta,\delta')}(p:q)= \chi_{\delta'}(F(p):F(q)) - \chi_\delta(p:q)F'(p), 
\end{equation}
where $F$ is a $(L_\delta,L_{\delta'})$-convex function.

\end{document}